\documentclass[pdfa, a4paper, USenglish]{lipics-v2021}

\usepackage{mathtools}  

\usepackage{pifont}  
\newcommand*{\cmark}{\ding{51}}
\newcommand*{\xmark}{\ding{55}}

\usepackage{fnpct}  

\usepackage{booktabs}

\allowdisplaybreaks[1]  

\renewcommand*{\emptyset}{\varnothing}  

\usepackage{optidef}

\DeclareMathOperator{\E}{\mathbb{E}}

\DeclareMathOperator{\Var}{Var}
\DeclareMathOperator{\Cov}{Cov}

\DeclarePairedDelimiterX\set[1]\lbrace\rbrace{\mkern1.5mu\def\suchthat{\;\delimsize|\;}#1\mkern1.5mu}

\usepackage{tikz}
\usetikzlibrary{babel, positioning, graphs}

\newcommand*{\Swap}{\textsf{Swap}}
\newcommand*{\Path}{\textsf{Path}}
\newcommand*{\ColorPath}{\textsf{ColorPath}}
\newcommand*{\BatchPath}{\textsf{BatchPath}}
\newcommand*{\Fence}{\textsf{Fence}}
\newcommand*{\DenseSwap}{\textsf{DenseSwap}}
\newcommand*{\DenseJump}{\textsf{DenseJump}}
\newcommand*{\ApproxStatic}{\textsf{ApproxStatic}}
\newcommand*{\ApproxDynamic}{\textsf{Approx}}

\newcommand*{\MinNodeMaxFlow}{\textsf{MinNodeMaxFlow}}

\newcommand*{\tree}{\mathcal{T}}
\newcommand*{\samples}{\mathcal{S}}
\newcommand*{\intervals}{\mathcal{I}}
\newcommand*{\eps}{\varepsilon}
\newcommand*{\epss}{\bar{\eps}}

\newlength{\capheight}
\newcommand{\necklace}[1]{%
    \StrLen{#1}[\strlen]%
    \foreach \x in {1,...,\strlen}{%
        \StrChar{#1}{\x}[\tmp]%
        \,\csname bead\tmp\endcsname\,%
    }%
}

\title{Dynamic Necklace Splitting}

\author{Rishi Advani}
{University of Illinois Chicago, IL, USA}
{radvani2@uic.edu}
{https://orcid.org/0000-0002-5522-0401}
{This author was supported by the UIC University Fellowship.}

\author{Abolfazl Asudeh}
{University of Illinois Chicago, IL, USA}
{asudeh@uic.edu}
{https://orcid.org/0000-0002-5251-6186}
{}

\author{Mohsen Dehghankar}
{University of Illinois Chicago, IL, USA}
{mdehgh2@uic.edu}
{https://orcid.org/0009-0006-1687-8012}
{}

\author{Stavros Sintos}
{University of Illinois Chicago, IL, USA}
{stavros@uic.edu}
{https://orcid.org/0000-0002-2114-8886}
{}

\authorrunning{R. Advani, A. Asudeh, M. Dehghankar, and S. Sintos}



\funding{This material is based upon work supported by the U.S.\ National Science Foundation under award No.~2348919.}

\acknowledgements{We would like to thank the reviewers for their constructive feedback.}

\hideLIPIcs
\nolinenumbers

\begin{document}


\maketitle

\begin{abstract}
    The necklace splitting problem is a classic problem in fair division with many applications, including data-informed fair hash maps. We extend necklace splitting to a dynamic setting, allowing for relocation, insertion, and deletion of beads. We present linear-time, optimal algorithms for the two-color case that support all dynamic updates. For more than two colors, we give linear-time, optimal algorithms for relocation subject to a restriction on the number of agents. Finally, we propose a randomized algorithm for the two-color case that handles all dynamic updates, guarantees approximate fairness with high probability, and runs in polylogarithmic time when the number of agents is small.
\end{abstract}

\section{Introduction}
The necklace splitting problem, first introduced by Bhatt and Leiserson~\cite{Bhatt1982}, is a classic combinatorial fair division problem. In this work, we extend it to a dynamic setting, enabling applications such as data-informed fair hash maps and improved load-balancing among multiple servers.

\subsection{Problem Setup}
We are given a string $S$ of $m$ \emph{beads}. We denote the $j$th element in $S$ by $S[j]$ and the substring from the $j_1$th element to the $j_2$th element by $[j_1, j_2]$\footnote{For ease of notation, when $j_1 > j_2$, we write $[j_1, j_2]$ instead of $[j_1, m] \cup [1, j_2]$.}. Each bead is a certain \emph{color} $i \in [n]$. Let $m_i = \lvert\set{b \in S \suchthat b\ \text{is color}\ i}\rvert$. Let $k$ be the number of agents. For simplicity, assume that $k$ divides $m_i$ for all $i$.
The goal is to find a set $P$ of small cardinality consisting of cuts of $S$ with the property that the resulting set of intervals can be allocated to the agents in such a way that each agent receives exactly $m_i/k$ beads of color $i$.
An overview of the notation is given in Table~\ref{tab:notation}.

\begin{table}[bpt]
    \centering
    \caption{Key notation used in the paper.}
    \label{tab:notation}
    \begin{tabular}{cl}
        \toprule
        Notation & Description \\
        \midrule
        $S$ & The string representing the necklace \\
        $P$ & The set of cuts in $S$ \\
        $m$ & The number of beads in $S$ \\
        $m_i$ & The number of beads of color $i$ in $S$ \\
        $n$ & The number of distinct colors of beads in $S$ \\
        $k$ & The number of agents \\
        $G$ & The neighborhood graph associated with $P$ \\
        $T$ & The neighborhood tree associated with $S$ and $P$ \\
        \bottomrule
    \end{tabular}
\end{table}

We are interested in the dynamic case. We are given an instance of the necklace splitting problem solved using an offline algorithm (see Section~\ref{sec:offline}). We allow the following dynamic updates:
\begin{description}
    \item[Relocation] Bead $S[j_1]$ is moved to index $j_2$.
    \item[Insertion] $\alpha k$ beads of the same color are inserted into $S$.
    \item[Deletion] $\alpha k$ beads of the same color are deleted from $S$.
\end{description}
We can also allow for recoloring of beads by performing successive deletion and insertion.

\subsection{Contributions}
Here we give a brief overview of our contributions:
\begin{itemize}
    \item We introduce and formalize the problem of dynamic necklace splitting.
    \item We design a linear-time algorithm for swapping adjacent beads when $n=2$ that achieves the optimal bound of $2(k-1)$ cuts. We design two linear-time algorithms for relocation of arbitrary distance when $n=2$ that achieve the optimal bound of $2(k-1)$ cuts. We also design a linear-time algorithm for relocation with looser bounds for general $n \geq 2$ (Section~\ref{sec:relocation_two_colors}).
    \item We introduce the \MinNodeMaxFlow{} problem and prove it to be NP-complete. We design an approximation algorithm for special cases that we then use as part of an algorithm to efficiently perform batch relocation when $n=2$ with the optimal number of cuts (Section~\ref{sec:relocation_batch}).
    \item We design two linear-time algorithms for relocation with general $n \geq 2$ when $m=nk$ that achieve the optimal bound of $n(k-1)$ cuts (Section~\ref{sec:relocation_multiple_colors}).
    \item We adapt our algorithm for batch relocation to be used for efficient insertion and deletion (Section~\ref{sec:insertion_and_deletion}).
    \item We design a randomized, polylogarithmic-time algorithm for relocation, insertion, and deletion that produces an approximately fair set of cuts with high probability when $n=2$ (Section~\ref{sec:approximate}).
\end{itemize}

See Table~\ref{tab:results} for a comparison of our algorithms.

\begin{table}[bhpt]
    \centering
    \caption{Key details of dynamic algorithms presented in the paper.}
    \label{tab:results}
    \begin{tabular}{ccccc@{}c}
        \toprule
        Algorithm & $n$ & Update & Exact? & Optimal \# of cuts? & Running time (per bead) \\
        \midrule
        \Swap{} & $2$ & Swap & \cmark & \cmark & $O\Bigl(\frac{m}{k}\Bigr)$ \\
        \Path{} & $2$ & Any & \cmark & \cmark & $O\Bigl(k + \frac{k'm}{k}\Bigr)$ \\
        \ColorPath{} & $2$ & Any & \cmark & \cmark & $O\Bigl(k + \frac{k'm}{k}\Bigr)$ \\
        \Fence{} & Any & Relocation & \cmark & \xmark & $O\Bigl(\frac{m}{kn}\Bigr)$ \\
        \BatchPath{} & $2$ & Any & \cmark & \cmark & $O\Bigl(\log k + \frac{k'm}{km'}\Bigr)$ \\
        \DenseSwap{} & $m/k$ & Swap & \cmark & \cmark & $\vphantom{\Bigl(}O(n)$ \\
        \DenseJump{} & $m/k$ & Relocation & \cmark & \cmark & $\vphantom{\Bigl(}O(k + n)$ \\
        \ApproxDynamic{} & $2$ & Any & \xmark & N/A & $\vphantom{\Bigl(}O\bigl(k^2 2^{2k} \eps^{-2} (\log m)^2 + \log m\bigr)$ \\
        \bottomrule
    \end{tabular}
\end{table}

\subsection{Data Structures}
We assume a standard RAM model in which the basic objects manipulated by the algorithms (e.g., beads) occupy a constant number of memory words and the entire input fits in main memory.

The choice of which data structures to use to store information about the necklace has a significant impact on the running time of our algorithms. We implement the necklace itself as a doubly linked list, allowing for efficient dynamic updates. Each node/bead also stores its index, the agent to whom it belongs, and a pointer to the next node/bead belonging to the same agent. Given any bead, this allows us to find the subsequence of beads belonging to the same agent in $O(m/k)$ time.

Note that the cuts are stored implicitly and can be explicitly generated in $O(m)$ time by iterating through the necklace and identifying pairs of consecutive beads where the associated agents switch. Alternatively, with minimal added cost, we can maintain a hash table mapping each possible pair of agents to the set of cuts adjacent to both agents.

For certain algorithms, we need to make use of a \emph{neighborhood graph} $G$ where each vertex represents an agent and two vertices are joined by an edge if the corresponding agents possess adjacent beads. We implement $G$ as an adjacency list. If the beads corresponding to $k'$ agents are reassigned among the same agents, $G$ can be updated in $O(k'm/k)$ time. We remove all edges incident to the $k'$ agents and then determine which edges to add by iterating through the $k'm/k$ beads assigned to those agents.

\subsection{Related Work}
Soon after the necklace splitting problem was originally introduced~\cite{Bhatt1982}, Goldberg and West~\cite{Goldberg1985} proved that a solution always exists if $k=2$. Alon and West~\cite{Alon1986} gave a simpler proof using the Borsuk--Ulam theorem~\cite{Borsuk1933}, and Alon~\cite{Alon1987} generalized the results to $k>2$. Finally, Alon and Graur~\cite{Alon2021} discovered an efficient approximation algorithm for finding a solution with few cuts.
Alon and Graur~\cite{Alon2021} also consider an online variation of the problem. They derive lower and upper bounds on the number of cuts needed when $k=2$ based on the value of $n$. They also generalize some of these results to the $k>2$ case.

In addition to necklace splitting, other fair division problems can be extended to dynamic settings. Kash~et~al.~\cite{Kash2014} study a fair division problem with divisible goods where agents with Leontief preferences arrive over time and goods must be irrevocably allocated. Benade~et~al.~\cite{Benade2018} study the problem of allocating indivisible goods that arrive in an online manner, again with irrevocable decisions. He~et~al.~\cite{He2019} generalize this to a setting where reallocation is allowed but expensive.

\subsection{Applications}
In this section, we discuss several key applications of the dynamic necklace splitting problem.

\subsubsection{Fair Hash Maps}
A core application of dynamic necklace splitting is designing practical fair hash maps. In addition to ensuring fairness, such data-informed hash maps can even be faster than traditional hash maps~\cite{Kraska2018, Sabek2022}. Shahbazi~et~al.~\cite{Shahbazi2024} use the static necklace splitting problem to design hash maps that satisfy group fairness. However, making dynamic updates to these hash maps is infeasible without a dynamic solution to the necklace splitting problem.

For a concrete application, consider the use of fair hash maps to maintain user information on a social network. Whenever new users join the network, their data needs to be efficiently added to the hash map, and when users leave, their data needs to be efficiently removed. As users' attributes change, their data needs to be efficiently updated (corresponding to relocation of beads in a necklace).

Another concrete application involves table joins in data lakes. Consider an organization seeking to share its data with third parties. The data is stored in a data lake and includes sensitive information as primary keys for certain tables. To protect user privacy, the data needs to be hashed, which introduces the risk of hash collision. A fair hashing scheme must be used to ensure any resulting errors when joining on the hashed columns do not disproportionately impact a specific demographic group~\cite{Shahbazi2024}. For this system to be efficient, it needs to be easy to update the hashes as the data changes.

\subsubsection{Load-Balancing}
Consider the problem of load-balancing among multiple servers, where tasks need to be completed in a specific order. Each agent corresponds to a server, and each bead color corresponds to a different type of task. We want to spread the load across the servers evenly, but we also want to minimize communication costs~\cite{JafarnejadGhomi2017}. This necessitates having as few cuts as possible.

If the order of tasks needs to be modified, new tasks need to be added, or existing tasks need to be canceled, we would want to be able to update the computation plan without recreating it from scratch.
As a special case, if computation has already begun, we can still perform dynamic updates that don't affect the already-completed tasks using dynamic necklace splitting.

\subsubsection{Bucketization}
Another key application is ensuring minority representation in bucketization. For example, when partitioning data for machine learning tasks, it is crucial that the training data is accurately represented by the test data for reliable results~\cite{Shahbazi2023}. Intentionally designing the bucketization process to maintain fairness helps prevent downstream fairness issues~\cite{Mehrabi2021, Pessach2022}. Here, each agent represents a bucket, and the color of each bead represents the grouping attribute. If the ``red'' data is more sparse than the ``blue'' data, it is important that the relative proportions are preserved to have accurate results.

\section{Offline Algorithm}
\label{sec:offline}
We present a simplified version of the algorithm of Shahbazi~et~al.~\cite{Shahbazi2024} here, as it is used as a subroutine in our algorithms for the $n = 2$ case.
Without loss of generality, we refer to the colors as red and blue.
Initialize a doubly linked list $L$ such that, for each $j \in [0, m-1]$, $L[j]$ is the number of red beads in the substring $[j, j+m/k-1]$. Initialize a hash set $H$ that contains all indices $j$ where $L[j] = m_1/k$. This takes $O(m)$ time.
By the discrete intermediate value theorem, there is a substring with $m_1/k$ red beads and $m_2/k$ blue beads, so $H$ is nonempty. We remove the smallest index $j$ from $H$ and allocate the sublist of length $m/k$ beginning at $j$ in $L$ to the first agent. We remove all indices within $m/k-1$ beads of $j$ in $L$ from $H$, update the values of $L$ for $m/k-1$ beads preceding $j$, add new indices to $H$ as necessary, and update the element preceding $j$ in $L$ to point to the next unallocated bead.
We repeat this process $k-1$ more times to allocate the remaining beads. Each step takes $O(m/k)$ time, so in total, the algorithm takes $O(m)$ time.
When we use this algorithm as a subroutine, we will often run it on a subsequence of length $k' m/k$ with only $k'$ agents. For this special case, the algorithm takes only $O(k' m/k)$ time.

As shown by Alon and Graur~\cite{Alon2021}, this algorithm produces at most $2(k-1)$ cuts.
It is also known that no algorithm can guarantee fewer than $2(k-1)$ cuts\footnote{See Appendix~\ref{sec:offline_optimal} for proof.}.

We now walk through a sample execution of the offline algorithm.

\begin{example}
    Suppose we are splitting the following necklace between three agents.
    \begin{center}
        \necklace{RRBRRBBBRBRB}
    \end{center}
    Each agent needs to receive two red beads and two blue beads. The first interval with the correct number of beads of each color is assigned to the first agent.
    \begin{center}
        \necklace{RRcBRRBcBBRBRB}
    \end{center}
    The second agent receives the first, second, seventh, and eighth beads, and the third agent receives the remaining beads.
    \begin{center}
        \necklace{RRcBRRBcBBcRBRB}
    \end{center}
    Each agent receives two red beads and two blue beads in total.\lipicsEnd  
\end{example}

\section{Relocation with Two Colors}
\label{sec:relocation_two_colors}
We now study the dynamic update of relocation with two colors.

\subsection{Adjacent Indices}
First, we consider the case where we restrict ourselves to relocating a bead to an adjacent index. Notice that this is equivalent to swapping two adjacent beads. For brevity, we will refer to relocation from $j$ to $j+1$ as swapping beads $S[j]$ and $S[j+1]$.

We now show how we can maintain a valid set of cuts after swapping two adjacent beads $S[j]$ and $S[j+1]$. If they are the same color or belong to the same agent, we can maintain the same set of cuts $P$. The only interesting case is when $S[j]$ and $S[j+1]$ are different colors and belong to different agents. Let $A_1$ be the owner of $S[j]$ and $A_2$ the owner of $S[j+1]$. Consider the subsequence of beads $S' \subseteq S$ belonging to either $A_1$ or $A_2$. We remove from $P$ the cuts adjacent to both $A_1$ and $A_2$. We run the offline algorithm on $S'$ (in $O(m/k)$ time) and add the (at most two) new cuts to $P$. We will henceforth refer to this procedure as \Swap{}.


Next, we prove upper bounds on the number of cuts in $P$ after the update. We start with a relatively trivial bound.

\begin{proposition}
\label{prop:swap_cuts_1}
    After $r$ swaps, \Swap{} produces a set of cuts of size at most $2(k-1) + r$.
\end{proposition}
\begin{proof}
    During each update, there is at least one cut adjacent to both $A_1$ and $A_2$. We remove this cut and add at most two cuts. Thus, with each update, we add at most one new cut. After $r$ swaps, we will have at most $2(k-1) + r$ cuts.
\end{proof}

In fact, we can show that no extra cuts are needed.

\begin{theorem}
\label{thm:swap_cuts_3}
    \Swap{} produces a set of cuts of size at most $2(k-1)$ and has time complexity $O(m/k)$.
\end{theorem}
\begin{proof}
    For the first update, we can assume that the original allocation has the structure of one given by an execution of the offline algorithm (up to order of interval selection). In addition, we will show that each update preserves that invariant, allowing us to use those properties in the proof of the bound.

    Assume without loss of generality that $A_1$ is allocated beads before $A_2$.
    Notice that the order in which agents are assigned beads can, in some cases, be altered without affecting the resulting allocations. In particular, since $A_1$ shares a common boundary with $A_2$, there can be no agent whose allotted set of beads encloses that of $A_1$ but not $A_2$. Thus, we can conceptually alter the order of execution of the original offline algorithm without loss of generality such that $A_1$ was allocated beads immediately before $A_2$. The only changes we have to make are in the cases where a cut is added during $A_1$'s turn that is adjacent to an agent $A_3$ whose beads are allocated after $A_1$ but before $A_2$. In those cases, we add the cut during $A_3$'s turn instead.

    Next, notice that each agent has two cuts allotted in the offline bound of $2(k-1)$ cuts (with the exception of the final agent, who has none). Thus, during each update, we can simply remove the cuts corresponding to $A_1$ (resulting in at most $2(k-2)$ total cuts) and add new cuts by rerunning the offline algorithm for the beads originally allocated to $A_1$ and $A_2$ (returning to the bound of $2(k-1)$). The only special case is when $A_1$ is not enclosed within $A_2$ and has no other adjacent agent at the time of allocation. In that case, we only remove one of two cuts corresponding to $A_1$, but we may add two cuts back. However, if we add two cuts, we can note that $A_2$ only has one corresponding cut, so $A_1$ can ``donate'' one cut to $A_2$, maintaining the property that each non-final agent has at most two corresponding cuts.
    
    Finally, we note that none of the steps in the update process violate the stated invariant that the allocation is one that could have been generated by an execution of the offline algorithm (up to order of interval selection), concluding our proof.
\end{proof}

\subsection{Nonadjacent Indices}
Next, we consider the case where beads need to be relocated an arbitrary distance away. Let $A_1$ be the owner of $S[j_1]$ and $A_2$ the owner of $S[j_2]$.

\subsubsection{The \Path{} Algorithm}
If the distance between $A_1$ and $A_2$ in the neighborhood graph $G$ is sufficiently small, we can efficiently perform relocations without adding any extra cuts. First, we find a shortest path between $A_1$ and $A_2$ in $G$. Let $k'$ be the number of agents along that path. Then, we move $S[j_1]$ to index $j_2$. Finally, we rerun the offline algorithm on the substring belonging to the $k'$ agents. We will henceforth refer to this procedure as \Path{}.

\begin{theorem}
\label{thm:path_cuts}
    \Path{} produces a set of cuts of size at most $2(k-1)$ and has time complexity $O(k + k'm/k)$.
\end{theorem}
\begin{proof}
	Note that the number of edges in $G$ is at most the number of cuts. Thus, finding a shortest path takes $O(k)$ time using breadth-first search, giving a running time of $O(k + k' m/k)$.
	 
    For the first update, we can assume that the original allocation has the structure of one given by an execution of the offline algorithm (up to order of interval selection). In addition, we will show that each update preserves that invariant, allowing us to use those properties in the proof of the bound.

    Each agent $A'$ in the path shares a common boundary with the next agent in the path, $A''$, so there can be no agent whose allotted set of beads ``encloses'' that of $A'$ but not $A''$. Without loss of generality, this enables us to conceptually alter the order of execution of the original offline algorithm such that $A'$ and $A''$ were assigned beads in consecutive turns. Following the same logic, we can consider all of the $k'$ agents to have had consecutive turns.

    Next, notice that each agent has two cuts allotted in the offline bound of $2(k-1)$ cuts (with the exception of the final agent, who has none). Thus, during each update, we can simply remove the cuts corresponding to the $k'$ agents (resulting in at most $2(k-k')$ total cuts) and add new cuts by rerunning the offline algorithm for the beads originally allocated to those agents (returning to the bound of $2(k-1)$).
    
    Finally, we note that none of the steps in the update process violate the stated invariant that the allocation is one that could have been generated by an execution of the offline algorithm (up to order of interval selection), concluding our proof.
\end{proof}

In many practical applications, the distance that beads need to be relocated is relatively small, making \Path{} very efficient. We now show that if relocation positions are drawn uniformly at random, we can bound the length of a path in $G$ with high likelihood.

\begin{lemma}
	\label{lemma:linear_graph}
	The average distance between a node $u$ and each other node in a connected graph with $k$ nodes is at most $k/2$.
\end{lemma}
\begin{proof}
	Let $d$ be the greatest distance between $u$ and any other node. The existence of such a node implies the existence of nodes at distances $d-1$, $d-2$, etc. from $u$ as well. The remaining $k-1-d$ nodes are at distance at most $d$ from $u$ by assumption. The average distance between $u$ and each other node can be bounded as follows.
	\[\frac{\sum_{\kappa=1}^d \kappa + (k-1 - d) d}{k-1} \leq \frac{\sum_{\kappa=1}^d \kappa + \sum_{\kappa=d+1}^{k-1} \kappa}{k-1} = \frac{\sum_{\kappa=1}^{k-1} \kappa}{k-1} = \frac{k(k-1)}{2(k-1)} = \frac{k}{2}\qedhere\]
\end{proof}

\begin{proposition}
	\label{prop:linear_graph}
	The distance between two distinct nodes in $G$ selected uniformly at random is at most $\lceil k/2 \rceil$ with probability at least 75\%.
\end{proposition}
\begin{proof}
	First, we show that the connected graph with maximum average path length is the linear graph. Note that the average path length of any graph is at most that of one of its spanning trees, so we can restrict ourselves to trees. We prove the claim by induction.
	
	For the base case, our claim trivially holds for $k < 3$. Assume the claim holds for some $k = \mu$. We show that it holds for $k = \mu + 1$. Consider the linear graph with $\mu + 1$ nodes. The subgraph induced by the first $\mu$ nodes has maximum average path length by assumption. The average distance between the last node and each other node is $k/2$, so by Lemma~\ref{lemma:linear_graph}, it cannot be increased. Therefore, the claim holds for $k = \mu + 1$, and by induction, it holds for all $k$.
	
	Next, we verify that it is possible for the neighborhood graph to be linear. One possible necklace that gives rise to a linear neighborhood graph is the following.
	
	\begin{center}
		\necklace{RBcRBcRBcRBcRBcRB}
	\end{center}
	
	Finally, we analyze the lengths of shortest paths in a linear graph. Suppose $G$ is a linear graph and two (distinct) nodes are selected uniformly at random. The distance between them follows a triangular distribution. The length is at most $\lceil k/2 \rceil$ with probability at least
	\[1 - \frac{\sum_{\kappa=1}^{k/2-1} \kappa}{\sum_{\kappa=1}^{k-1} \kappa} = 1 - \frac{k(k-2)}{4k(k-1)} = \frac{3k-2}{4k-4} > \frac{3k-3}{4k-4} = 75\% \,.\qedhere\]
\end{proof}


\subsubsection{The \ColorPath{} Algorithm}
Without loss of generality, assume $S[j_1]$ is red.
We now construct a more intricate neighborhood graph, $G'$. Unlike $G$, this graph is directed and weighted. For each pair of agents $A'$ and $A''$, there is an edge from $A'$ to $A''$ if the two agents possess adjacent beads. If any such pair of beads has a red bead on the side of $A'$, we call the edge \emph{good}, and its weight is $0$; otherwise, we call it \emph{bad}, and its weight is $1$.


First, we find a shortest path from $A_2$ to $A_1$ in $G'$. Let $k'$ be the number of agents along that path. Then, we move $S[j_1]$ to index $j_2$. Next, for each good edge, we move the corresponding red bead across the adjacent cut. For each subpath of bad edges, we rerun the offline algorithm on the substring corresponding to the subpath. We will henceforth refer to this procedure as \ColorPath{}.

Note that the number of edges in $G'$ is at most twice the number of cuts. Finding a shortest path takes $O(k)$ time using Dial's algorithm~\cite{Dial1969}, giving a running time of $O(k + k' m/k)$ for \ColorPath{}. Asymptotically, this is the same as \Path{}, but in practice, \ColorPath{} will likely be faster.


\begin{proposition}
\label{prop:colorpath_cuts}
    \ColorPath{} produces a set of cuts of size at most $2(k-1)$ and has time complexity $O(k + k'm/k)$.
\end{proposition}

\subsubsection{The \Fence{} Algorithm}
If the number of agents $k'$ in the path is large, we may instead wish to add extra cuts in return for a reduced running time. After moving $S[j_1]$ to index $j_2$, instead of rerunning the offline algorithm, we simply add (at most two) cuts around it as necessary. We will henceforth refer to this procedure as \Fence{}.

The main disadvantage of \Fence{} is that the resulting allocation is not guaranteed to be one that could have been generated by an execution of the offline algorithm. This means that, after a single execution of \Fence{}, the guarantees of Theorems~\ref{thm:swap_cuts_3}~and~\ref{thm:path_cuts} (and Proposition~\ref{prop:colorpath_cuts}) no longer hold\footnote{Proposition~\ref{prop:swap_cuts_1} can technically still be made to work with some adjustments. After usage of \Fence{}, a modified Proposition~\ref{prop:swap_cuts_1} would guarantee that the number of cuts increases by one with each execution of \Swap{}.}. As such, once \Fence{} has been used, any further relocations must be performed using \Fence{} as well. Furthermore, once the guaranteed number of cuts has passed some user-determined tolerance level, the whole necklace should be reallocated from scratch using the offline algorithm.

Suppose $2k$ extra cuts are permitted to be added via relocations before reallocation from scratch is required. Consider the sequence of relocations starting from the first execution of \Fence{} to the first relocation that leads to over $2k$ extra cuts. The amortized running time of \Fence{} (including the time needed for the reallocation from scratch) is then
\[O(1) + \frac{O(m)}{\Omega(k)} = O\biggl(\frac{m}{k}\biggr) \,.\]
As such, \Fence{} may be preferred over \Path{} and \ColorPath{} if $k'$ is frequently large.

We can also adapt \Fence{} to the case of $n > 2$ in a straightforward way. Using the offline algorithm of Alon and Graur~\cite[Theorem~5]{Alon2021} (which has a running time of $O(m)$), we can guarantee at most $n(k-1)\lceil 4+\log_2(3k \max_{i \in [n]} m_i) \rceil + 2r$ cuts after $r$ relocations. If $2kn$ extra cuts are permitted to be added before reallocation from scratch is required, the amortized running time is then
\[O(1) + \frac{O(m)}{\Omega(kn)} = O\biggl(\frac{m}{kn}\biggr) \,.\]

\begin{proposition}
\label{prop:fence_cuts}
    After $r$ relocations, \Fence{} produces a set of cuts of size at most $2(k+r-1)$ if $n=2$ and $n(k-1)\lceil 4+\log_2(3k \max_{i \in [n]} m_i) \rceil + 2r$ if $n>2$ and has time complexity $O\Bigl(\frac{m}{kn}\Bigr)$.
\end{proposition}

\section{Batch Relocation with Two Colors}
\label{sec:relocation_batch}
To make relocation more efficient, we can perform batch updates. Instead of moving a single bead, we will move $m'$ beads (of the same color)\footnote{Our techniques can be adapted in a straightforward manner to work with beads of both colors being relocated simultaneously, but for ease of exposition, we restrict ourselves to a single color here.} from indices $j_1^1, \dots, j_1^{m'}$ to indices $j_2^1, \dots, j_2^{m'}$. Let $\mathbf{A_1}$ and $\mathbf{A_2}$ be the corresponding sets of agents. If there are any agents in both sets, we can remove them from both. Let $k''$ be the number of agents remaining. Let $m''$ be the number of beads that need to change owners.

We construct a flow network $G'' = (V'', E'')$ from $G$ as follows. For each node in $V$, there is a corresponding node in $V''$. In addition, there is a source $s$ and sink $t$. For each node $u \in V'' \setminus \set{s, t}$, let $\Delta(u)$ denote the net change in the number of beads that the corresponding agent owns. For each node $u$ corresponding to an agent in $\mathbf{A_1}$, there is an edge $(s, u)$ with capacity $c_{su} \coloneqq -\Delta(u)$. For each node $u$ corresponding to an agent in $\mathbf{A_2}$, there is an edge $(u, t)$ with capacity $c_{ut} \coloneqq \Delta(u)$. For each edge in $E$, there is a corresponding edge $(u, v)$ in $E''$ with infinite capacity.  

Define an \emph{active} node to be one with positive incoming and outgoing flow. We want to find a max flow (of value $m''$) such that the number of active nodes is minimized. We will henceforth refer to this problem as \MinNodeMaxFlow{}.

\begin{proposition}
\label{prop:minnodemaxflow}
    \MinNodeMaxFlow{} is NP-complete\footnote{See Appendix~\ref{sec:minnodenaxflow} for proof.}.
\end{proposition}

We now describe the construction of a spanning tree of $G$ that will allow us to approximately solve \MinNodeMaxFlow{}.
Consider the process of generating the neighborhood graph $G$ while the offline necklace splitting algorithm is being run. When a new node is added, we initially assign it a \emph{level} of 1. If the corresponding agent's interval encloses any other agents' intervals, then we increment the levels of all the nodes corresponding to agents with enclosed intervals.

For example, the node corresponding to the first agent allocated beads is initially on level 1. If the second agent's interval encloses that of the first, the first agent's node is moved to level 2. If the third agent's interval encloses that of the second, then we have the first agent's node on level 3 and the second agent's node on level 2.

We continue the process of updating the levels until all agents have been assigned their intervals and the graph $G$ is finalized. Then, we construct a new graph by removing all edges between nodes on the same level, other than those on level 1. We will denote this \emph{neighborhood tree} by $T$.

\begin{figure}[bhpt]
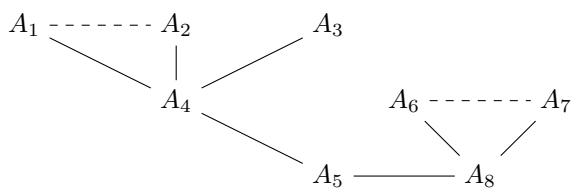

    \centering
    \tikz \graph [no placement] {
        {[y=2]
            a1/$A_1$[x=0] --[dashed] a2/$A_2$[x=2],
            a3/$A_3$[x=4]
        },
        {[y=1]
            a4/$A_4$[x=2] -- a1,
            a4 -- a2,
            a4 -- a3,
            a6/$A_6$[x=5] --[dashed] a7/$A_7$[x=7]
        },
        {[y=0]
            a5/$A_5$[x=4] -- a4,
            a8/$A_8$[x=6] -- a6,
            a8 -- a7,
            a5 -- a8
        }
    };
    \caption{An example of a neighborhood tree. The dashed lines indicate additional edges present in the corresponding neighborhood graph.}
\end{figure}

\begin{proposition}
    $T$ is a spanning tree of $G$.
\end{proposition}
\begin{proof}
    First, we prove that $T$ is connected. $G$ is trivially connected since the necklace it represents is a contiguous collection of beads. We need to show that the edges removed from $G$ to construct $T$ are unnecessary for connectivity. Consider any edge removed. By construction, neither endpoint is on level 1. Thus, both endpoints are joined to a node on the previous level corresponding to an enclosing agent, and consequently, they remain connected.

    Next, we prove that $T$ is acyclic. Assume for contradiction that there exists a cycle in $T$. Let $\ell$ be the lowest level represented in the cycle. Consider a node $u$ in the cycle on level $\ell$. Both of its neighbors in the cycle must be on level $\ell - 1$. However, by construction, each node is joined to at most one node on the level below it. We have a contradiction, so $T$ must be acyclic.
    Therefore, $T$ is a spanning tree of $G$.
\end{proof}


We now describe how we can efficiently produce an approximate solution to \MinNodeMaxFlow{} using the properties of $T$.
We construct a flow network $T'$ from $T$ analogous to $G''$. The only difference from $G''$ is the lack of edges between nodes on the same level (other than level 1).

\begin{lemma}
    The number of active nodes in a solution to \MinNodeMaxFlow{} on $T'$ is at most twice the number of active nodes in a solution on $G''$.
\end{lemma}
\begin{proof}
    Consider an optimal solution to \MinNodeMaxFlow{} on $G''$. Assume, without loss of generality, that there is no pair of nodes $u$ and $v$ such that both $(u,v)$ and $(v,u)$ have positive flow. Let $k'$ be the number of active nodes. For each edge $(u,v)$ with positive flow where $u$ and $v$ are on the same level (other than level 1), we can remove the flow and instead add equal flow to $(u,w)$ and $(w,v)$, where $w$ is the unique node on the previous level joined to both $u$ and $v$. If we replicate the resulting flow on $T'$, we have a feasible solution to \MinNodeMaxFlow{} on $T'$.

    Since each node is joined to at most two nodes on the same level, there are at most $k'$ pairs of adjacent active nodes on the same level. In the worst case, we add $k'$ active nodes to the solution, resulting in $2k'$ active nodes.
    The optimal solution to \MinNodeMaxFlow{} on $T'$ must have at most as many active nodes as this feasible solution, so it also has at most $2k'$ active nodes. Therefore, the number of nodes is at most twice the number of active nodes in an optimal solution on $G''$.
\end{proof}

\begin{lemma}
    \MinNodeMaxFlow{} can be solved in $O(k'' \log k'' + k')$ time on $T'$.
\end{lemma}
\begin{proof}
    The edges from the source $s$ and to the sink $t$ must all be saturated to produce a max flow, so we can ignore both nodes. The remaining subgraph has structure identical to $T$, and is thus a tree.
    We sort the sources/sinks by their level in descending order and initialize a linked list of pointers to them in that order.
    Consider the set of nodes with excess (potentially negative) flow on the highest level (other than level 1). At first, these are all sources/sinks. Each of these nodes has a unique edge leading to the level below. Any valid flow must have flow on these edges. We place the appropriate amount of flow on each of these edges to satisfy the demands of the sources/sinks. Then we consider the new set of nodes with excess flow on the highest level (other than level 1) and repeat the process, updating the linked list accordingly. We terminate once all nodes not on level 1 have no excess flow. In each step, the highest level with excess flow decreases, so this process must terminate.

    If any nodes with excess flow remain, they are on level 1. Consider the leftmost such node. We put flow on each edge to the right until the current node has no excess flow. Then we find the new leftmost node with excess flow and repeat the process. In each step, the number of nodes with excess flow decreases, so eventually no nodes will have excess flow, at which point we terminate.

    Sorting the initial list takes $O(k'' \log k'')$ time. In the first phase, we process only the nodes in the solution, and each in constant time, giving a running time of $O(k')$ for the first phase. In the second phase, assuming the relative position of each node on level 1 is stored when $T$ is constructed, we can find the leftmost node in $O(k')$ time, and after that, we only process nodes in the solution, each in constant time. This gives us a running time of $O(k')$ for the second phase. Overall, we have a running time of $O(k'' \log k'' + k')$ for solving \MinNodeMaxFlow{} on $T'$.
\end{proof}

\begin{example}
Consider the following neighborhood tree.
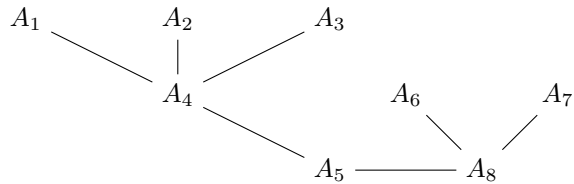
\begin{center}
    \tikz \graph [no placement] {
        {[y=2]
            a1/$A_1$[x=0],
            a2/$A_2$[x=2],
            a3/$A_3$[x=4]
        },
        {[y=1]
            a4/$A_4$[x=2] -- a1,
            a4 -- a2,
            a4 -- a3,
            a6/$A_6$[x=5],
            a7/$A_7$[x=7]
        },
        {[y=0]
            a5/$A_5$[x=4] -- a4,
            a8/$A_8$[x=6] -- a6,
            a8 -- a7,
            a5 -- a8
        }
    };
\end{center}
Suppose $A_1$ and $A_5$ are sources and $A_3$ and $A_6$ are sinks, all with demands of 1. At first, $A_1$ and $A_3$ are the highest nodes with excess flow. We put flow on the edges $(A_1, A_4)$ and $(A_4, A_3)$.
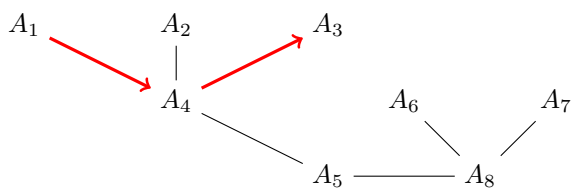
\begin{center}
    \tikz \graph [no placement] {
        {[y=2]
            a1/$A_1$[x=0],
            a2/$A_2$[x=2],
            a3/$A_3$[x=4]
        },
        {[y=1]
            a4/$A_4$[x=2] <-[red, very thick] a1,
            a4 -- a2,
            a4 ->[red, very thick] a3,
            a6/$A_6$[x=5],
            a7/$A_7$[x=7]
        },
        {[y=0]
            a5/$A_5$[x=4] -- a4,
            a8/$A_8$[x=6] -- a6,
            a8 -- a7,
            a5 -- a8
        }
    };
\end{center}
Now the highest node with excess flow is $A_6$, so we put flow on $(A_8, A_6)$.
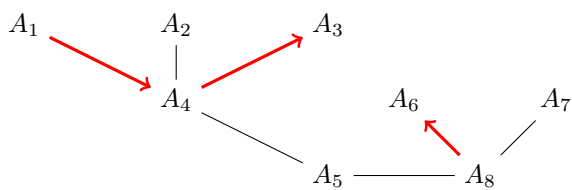
\begin{center}
    \tikz \graph [no placement] {
        {[y=2]
            a1/$A_1$[x=0],
            a2/$A_2$[x=2],
            a3/$A_3$[x=4]
        },
        {[y=1]
            a4/$A_4$[x=2] <-[red, very thick] a1,
            a4 -- a2,
            a4 ->[red, very thick] a3,
            a6/$A_6$[x=5],
            a7/$A_7$[x=7]
        },
        {[y=0]
            a5/$A_5$[x=4] -- a4,
            a8/$A_8$[x=6] ->[red, very thick] a6,
            a8 -- a7,
            a5 -- a8
        }
    };
\end{center}
Continuing, we put flow on $(A_5, A_8)$.
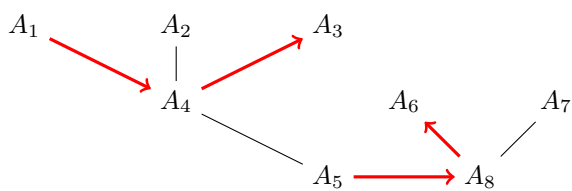
\begin{center}
    \tikz \graph [no placement] {
        {[y=2]
            a1/$A_1$[x=0],
            a2/$A_2$[x=2],
            a3/$A_3$[x=4]
        },
        {[y=1]
            a4/$A_4$[x=2] <-[red, very thick] a1,
            a4 -- a2,
            a4 ->[red, very thick] a3,
            a6/$A_6$[x=5],
            a7/$A_7$[x=7]
        },
        {[y=0]
            a5/$A_5$[x=4] -- a4,
            a8/$A_8$[x=6] ->[red, very thick] a6,
            a8 -- a7,
            a5 ->[red, very thick] a8
        }
    };
\end{center}
All nodes are now satisfied, so the algorithm terminates.\lipicsEnd
\end{example}

After \MinNodeMaxFlow{} has been solved on $T'$, we can return to $G$ to perform a heuristic optimization.
For each active node $u$ that is neither a source/sink nor sending flow to the level below, consider the subgraph consisting of its active neighbors on the level above. If that subgraph is connected, then we can remove $u$ from the solution. Finding all such subgraphs takes $O(k)$ time since we consider each edge a constant number of times. Checking if the subgraphs are connected takes $O(k)$ time in total. Overall, this heuristic pruning process adds $O(k)$ additional time.

Finally, we rerun the offline necklace splitting algorithm on the substring belonging to the agents corresponding to the active nodes. We will henceforth refer to this procedure as \BatchPath{}.
We have an overall running time of $O\Bigl(\frac{1}{m'} \Bigl(k'' \log k'' + k' + \frac{k'm}{k}\Bigr)\Bigr) = O\Bigl(\frac{k'' \log k''}{m'} + \frac{k'm}{km'}\Bigr)$ per bead for \BatchPath{} (without pruning). Since $k'' \leq m'$, we also have the looser but more intuitive running time of $O\Bigl(\log k + \frac{k'm}{km'}\Bigr)$. Thus, there is only an $O(\log k)$ overhead when using \BatchPath{} compared to \Path{}, without accounting for the $m'\times$ speedup from batching updates.


\begin{theorem}
    \BatchPath{} produces a set of cuts of size at most $2(k-1)$ and has time complexity $O\Bigl(\log k + \frac{k'm}{km'}\Bigr)$.
\end{theorem}

\section{Relocation with Multiple Colors and High Cut Density}
\label{sec:relocation_multiple_colors}
Next, we consider the case of relocation for general $n \geq 2$ with the additional restriction that $m = nk$. Note that using the corresponding offline algorithm of Alon and Graur~\cite[Proposition~1]{Alon2021}, the initial solution is guaranteed to have exactly $n(k-1)$ cuts (some may be redundant). This initial solution is produced in $O(m)$ time.
To enable efficient implementations for our dynamic algorithms, we maintain a two-dimensional array indexed by agent and color that points to the corresponding bead in $S$.

\subsection{Adjacent Indices}
Suppose we need to swap a red bead (on the left) and blue bead (on the right) between $A_1$ and $A_2$, respectively. There are three possible cases based on whether or not there is a cut to the left of the red bead and/or the right of the blue bead (Table~\ref{tab:denseswap}).

\begin{table}[bhpt]
    \centering
    \caption{The possible cases for the \DenseSwap{} algorithm.}
    \label{tab:denseswap}
    \begin{tabular}{cl}
        \toprule
        Case & Necklace \\
        \midrule
        1 & \necklace{cRcBc} \textrightarrow{} \necklace{cBcRc} \\[1ex]
        2 & \necklace{cRcBG} \textrightarrow{} \necklace{cBcRG} \\[1ex]
        3a & \necklace{BRcB} \textrightarrow{} \necklace{BBcR} \textrightarrow{} \necklace{BcBR} \\[1ex]
        3b & \necklace{GRcB} \textrightarrow{} \necklace{GBcR} \textrightarrow{} \necklace{GcBR} \\
        \bottomrule
    \end{tabular}
\end{table}

In the first case, there is a cut to the left of the red bead and to the right of the blue bead (or it is at the end of the necklace). This is the simplest case -- we swap the beads and let the agents retain ownership of their original beads. No cuts need to be adjusted.

In the second case, there is a cut on the left, but not the right. We swap the beads and do not adjust any cuts. $A_2$ takes possession of the left interval, and $A_1$ takes possession of the right interval. As part of the right interval, $A_1$ gains at least one unneeded extra bead (the green one). Each extra bead lacks a cut to its left, so they must be the first of their colors. Thus, the original beads of $A_1$ are somewhere to the right in the necklace. Furthermore, since each of the original beads is not the first of its color, they each start their own intervals. $A_2$ takes ownership of each of these intervals and repeats the process with each of them if they have multiple beads, continuing to attempt to re-satisfy the fairness constraint. Since the necklace has exactly $n-1$ pairs of adjacent beads without a cut, this results in at most $n$ ownership exchanges throughout the procedure.

In the third case, there is no cut on the left. We swap the beads and move the cut left by one bead (to ensure that the set of cuts maintains the structure of one generated by the offline algorithm). $A_1$ and $A_2$ retain the left and right intervals, respectively. Similarly to the second case, to maintain the fairness constraint, at most $n-1$ intervals will have to be exchanged between the two agents.

We will henceforth refer to this procedure as \DenseSwap{}.
The two-dimensional array allows us to perform each ownership exchange in constant time, giving a running time of $O(n)$ for \DenseSwap{}.

\begin{theorem}
    \DenseSwap{} produces a set of cuts of size exactly $n(k-1)$ and has time complexity $O(n)$.
\end{theorem}

\subsection{Nonadjacent Indices}
We can generalize the ideas from the previous section to design an algorithm for arbitrary relocation. In addition to the two-dimensional array, we maintain for each color a red--black tree tracking the indices of the beads of that color in $S$. Suppose we need to relocate a red bead. There are four possible cases.


In the first case, the bead to be relocated begins as the first of its color and remains so. We move the bead to the desired position, and if the interval is owned by a different agent, we re-satisfy the fairness constraint by exchanging at most $n$ intervals between the two agents, as in the previous section.

In the second case, the bead to be relocated begins not as the first of its color and remains so. We move the bead and its corresponding cut to an arbitrary boundary to make it a singleton interval, and if the two affected intervals are not owned by the same agent, we re-satisfy the fairness constraint by exchanging at most $n-1$ intervals between the two agents. Then, we move the bead and its corresponding cut to the desired position, and if the destination interval does not belong to the owner of the bead, we re-satisfy the fairness constraint by exchanging at most $n$ intervals between the two agents.

In the third case, the bead begins as the first of its color but does not remain so. We relocate the bead directly to the left of the second red bead via Case~1, and then we relocate the second red bead to the desired position via Case~2. At most $3n-1$ intervals are exchanged throughout this case.

In the fourth case, the bead begins not as the first of its color but becomes the first. We relocate the bead and its corresponding cut directly to the right of the first red bead via Case~2, and then we relocate the first red bead to the desired position via Case~1. At most $3n-1$ intervals are exchanged throughout this case.

We will henceforth refer to this procedure as \DenseJump{}.
Updating the red--black tree in Case~2 takes $O(k)$ time if the rank of the bead changes. Otherwise, all updates to the tree take constant time. This gives us a total running time of $O(k + n)$ for \DenseJump{}.

\begin{theorem}
    \DenseJump{} produces a set of cuts of size exactly $n(k-1)$ and has time complexity $O(k + n)$.
\end{theorem}

\section{Insertion and Deletion}
\label{sec:insertion_and_deletion}
In this section, we study the complementary dynamic updates of insertion and deletion with two colors.
For insertion, we simply give $\alpha$ beads to each agent and then relocate the $\alpha k$ beads to the desired indices using \BatchPath{}. For deletion, we delete the desired beads and then relocate beads from agents with surpluses to those with deficits, again using \BatchPath{}. Both insertion and deletion take $O\Bigl(\frac{k'' \log k''}{\alpha k} + \frac{k'm}{\alpha k^2}\Bigr) = O\Bigl(\frac{\log k}{\alpha} + \frac{k'm}{\alpha k^2}\Bigr)$ time per bead using this approach.

\begin{theorem}
    There is an algorithm for insertion and deletion that produces a set of cuts of size at most $2(k-1)$ and has time complexity $O\Bigl(\frac{\log k}{\alpha} + \frac{k'm}{\alpha k^2}\Bigr)$.
\end{theorem}

We now show how the running time can be more precisely understood in the special case of $\alpha = 1$ by bounding the value of $k''$ as $k$ grows.
Suppose each of the beads is associated with a scalar value, such that the position of a bead in the necklace is determined by the rank of its value\footnote{This is the case, for example, when designing fair hash maps~\cite{Shahbazi2024}.}. If the bead values are drawn independently from the same distribution $\mathcal{D}$, then the distribution of the rank of a new sample among the existing beads in the necklace is simply the uniform distribution. As such, it is well motivated to consider beads as being inserted into the necklace uniformly at random. Equivalently, agents to receive beads are chosen uniformly at random. In the absence of any further information, it is justified to assume that beads to be deleted are chosen uniformly at random as well.

Ideally (to reduce running time), we want to have as many agents receive exactly one bead as possible, since $k''$ is the number of agents that do not receive exactly one bead. We will show that this is likely to be the case. Assume $k > 101$. We sample $k$ values from the discrete uniform distribution on $[k]$. For each $i \in [k]$, let $X_i$ be the Bernoulli random variable indicating if the number $i$ is not chosen exactly once\textbackslash{}footnote\{Note that \$X\_i\$ and \$X\_j\$ are dependent random variables for \$i\$, \$j \textbackslash{}in [k]\$.\}. Let $X = \sum_i X_i$. We have $k(1 - 1.005/e) < \E[X] < k(1 - 1/e)$. For the variance, we have the following.
\begin{align*}
    \Var(X) &= \Var(k - X) \\
    &= \sum_i \Var(1 - X_i) + \smashoperator{\sum_{i,j=1, i \ne j}^k} \Cov(1 - X_i, 1 - X_j) \\
    &= k\Var(1 - X_1) + k(k-1) \Cov(1 - X_1, 1 - X_2) \\
    &< \frac{1.005k}{e} - \frac{k}{e^2} + k(k-1) \biggl(\frac{(k-1)(k-2)^{k-2}}{k^{k-1}} - \frac{(k-1)^{2k-2}}{k^{2k-2}}\biggr) \\
    &= \frac{1.005e - 1}{e^2} k + \frac{k^{k-1}(k-1)^2(k-2)^{k-2} - (k-1)^{2k-1}}{k^{2k-1}} \\
    &< \frac{1.005e - 1}{e^2} k + \frac{k^{2k-1} - (k-1)^{2k-1}}{k^{2k-1}} \\
    &= \frac{1.005e - 1}{e^2} k + 1 - \biggl(1 - \frac{1}{k}\biggr)^{2k-1} \\
    &< \frac{1.005e - 1}{e^2} k + 1 - \frac{0.995}{e^2}
\end{align*}

Then, by the Chebyshev--Cantelli inequality, we have
\[\Pr(X \geq \lambda k) < \frac{\frac{1.005e - 1}{e^2} k + 1 - \frac{0.995}{e^2}}{\frac{1.005e - 1}{e^2} k + 1 - \frac{0.995}{e^2} + \Bigl(\lambda k + \frac{1}{e} - 1\Bigr)^2} \,.\]
If we set $\lambda = 1 - 1/e$, then we have $k'' < (1-1/e)k$ with probability greater than 99.3\%.

\section{Approximate Necklace Splitting}
\label{sec:approximate}
In this section, we design an algorithm for approximate necklace splitting with two colors that works in both the static and dynamic settings. The algorithm is efficient when the number of agents $k$ is small, for example, $k=o(m)$.

\subsection{Static Setting}
We construct a red--black tree $\tree$ over the necklace $S$ with respect to the order of the beads (i.e., we assume that $S[j_1] < S[j_2]$ if $j_1 < j_2$). Let $S_1$ be the subsequence consisting of red beads and $S_2$ the subsequence consisting of blue beads with $|S_1|=m_1$ and $|S_2|=m_2$. Let $\eps \in (0,1)$ be a parameter specified by the user. Let $\intervals^0 = \emptyset$. We run the following subroutine for $k$ iterations.

In the $j$th iteration, for each color $i \in \set{1, 2}$, we use $\tree$ to obtain a uniform random sample $\samples_i^j \subseteq S_i \setminus \intervals^{j-1}$ of cardinality $O\bigl((k-j+1)^2 2^{2k} \eps^{-2} \log(2km)\bigr)$. Let $\samples^j = \samples_1^j \cup \samples_2^j$. We execute one iteration of the exact offline algorithm with $k-j+1$ agents on $S \setminus \intervals^{j-1}$. Let $I_j$ be the interval returned by the algorithm. We add the (at most) two corresponding cuts. Let $\intervals^j = \intervals^{j-1} \cup \set{I_j}$.

After the $k$th iteration, the resulting set of cuts is returned. We will henceforth refer to this procedure as \ApproxStatic{}.

\begin{theorem}
\label{thm:approxstatic}
\ApproxStatic{} produces an approximate solution with at most $2(k-1)$ cuts, such that each agent gets at least $(1-\eps)\frac{m_i}{k}$ and at most $(1+\eps)\frac{m_i}{k}$ beads of color $i$ with probability at least $1 - \frac{1}{m}$. It has time complexity $O\bigl(m+k^3 2^{2k}\eps^{-2}(\log m)^2\bigr)$\footnote{See Appendix~\ref{sec:approxstatic} for proof.}.
\end{theorem}

\subsection{Dynamic Setting}
We now show how to adapt \ApproxStatic{} to the dynamic setting. Our algorithm can handle any type of dynamization, including relocation, insertion, and deletion.

Notice that the crux of \ApproxStatic{} relies on a small set of samples (given sufficiently small $k$). Hence, if we ensure that sampling is performed efficiently under updates, we can have a dynamic algorithm for approximate necklace splitting with an update time that only depends logarithmically on $m$. Indeed, assume that the tree $\tree$ is constructed in the preprocessing phase. It is known that a red--black tree can be updated in $O(\log m)$ time under insertion/deletion of an element. Thus, our algorithm for the dynamic necklace splitting problem consists of maintaining $\tree$ under relocations/insertions/deletions (a relocation can be handled as a deletion followed by an insertion). When the user requests the set of cuts, we execute \ApproxStatic{} as a subroutine using the updated $\tree$. We will henceforth refer to this procedure as \ApproxDynamic{}.

When a batch of $k$ updates is encountered, $\tree$ is updated in $O(k \log m)$ time, and the new set of cuts can be constructed in $O\bigl(k^3 2^{2k} \eps^{-2} (\log m)^2\bigr)$ time. The overall running time of \ApproxDynamic{} is thus $O\bigl(k^2 2^{2k} \eps^{-2} (\log m)^2 + \log m\bigr)$ per bead. If we do not often need to produce the actual set of cuts, we only need $O(\log m)$ time per bead to maintain $\tree$.

\begin{theorem}
\ApproxDynamic{} produces an approximate solution with at most $2(k-1)$ cuts, such that each agent gets at least $(1-\eps) \frac{m_i}{k}$ and at most $(1+\eps) \frac{m_i}{k}$ beads of color $i$ with probability at least $1 - \frac{1}{m}$. It has time complexity $O\bigl(k^2 2^{2k} \eps^{-2} (\log m)^2 + \log m\bigr)$.
\end{theorem}

\bibliographystyle{plainurl}
\bibliography{main}

\begin{thebibliography}{10}

\bibitem{Alon1987}
Noga Alon.
\newblock Splitting necklaces.
\newblock {\em Advances in Mathematics}, 63(3):247--253, March 1987.
\newblock \href {https://doi.org/10.1016/0001-8708(87)90055-7}
  {\path{doi:10.1016/0001-8708(87)90055-7}}.

\bibitem{Alon2021}
Noga Alon and Andrei Graur.
\newblock Efficient splitting of necklaces.
\newblock In {\em 48th International Colloquium on Automata, Languages, and
  Programming (ICALP 2021)}, volume 198 of {\em ICALP '21}, pages 14:1--14:17,
  Virtual (Glasgow, Scotland), July 2021. Schloss Dagstuhl -- Leibniz-Zentrum
  f{\"u}r Informatik.
\newblock \href {https://doi.org/10.4230/LIPIcs.ICALP.2021.14}
  {\path{doi:10.4230/LIPIcs.ICALP.2021.14}}.

\bibitem{Alon1986}
Noga Alon and Douglas~B. West.
\newblock The {Borsuk}-{Ulam} theorem and bisection of necklaces.
\newblock {\em Proceedings of the American Mathematical Society},
  98(4):623--628, 1986.
\newblock \href {https://doi.org/10.1090/s0002-9939-1986-0861764-9}
  {\path{doi:10.1090/s0002-9939-1986-0861764-9}}.

\bibitem{Anthony1999}
Martin Anthony and Peter~L. Bartlett.
\newblock {\em Neural Network Learning: Theoretical Foundations}.
\newblock Cambridge University Press, November 1999.
\newblock \href {https://doi.org/10.1017/cbo9780511624216}
  {\path{doi:10.1017/cbo9780511624216}}.

\bibitem{Benade2018}
Gerdus Benade, Aleksandr~M. Kazachkov, Ariel~D. Procaccia, and
  Christos-Alexandros Psomas.
\newblock How to make envy vanish over time.
\newblock In {\em Proceedings of the 2018 ACM Conference on Economics and
  Computation}, EC ’18, pages 593--610, Ithaca, New York, USA, June 2018.
  ACM.
\newblock \href {https://doi.org/10.1145/3219166.3219179}
  {\path{doi:10.1145/3219166.3219179}}.

\bibitem{Bhatt1982}
Sandeep~N. Bhatt and Charles~E. Leiserson.
\newblock How to assemble tree machines (extended abstract).
\newblock In {\em Proceedings of the Fourteenth Annual ACM Symposium on Theory
  of Computing}, STOC ’82, pages 77--84, San Francisco, California, USA, May
  1982. {ACM}.
\newblock \href {https://doi.org/10.1145/800070.802179}
  {\path{doi:10.1145/800070.802179}}.

\bibitem{Borsuk1933}
Karol Borsuk.
\newblock {Drei} {S\"{a}tze} \"{u}ber die n-dimensionale euklidische
  {Sph\"{a}re}.
\newblock {\em Fundamenta Mathematicae}, 20(1):177--190, 1933.
\newblock \href {https://doi.org/10.4064/fm-20-1-177-190}
  {\path{doi:10.4064/fm-20-1-177-190}}.

\bibitem{Dial1969}
Robert~B. Dial.
\newblock Algorithm 360: shortest-path forest with topological ordering [h].
\newblock {\em Communications of the ACM}, 12(11):632--633, November 1969.
\newblock \href {https://doi.org/10.1145/363269.363610}
  {\path{doi:10.1145/363269.363610}}.

\bibitem{Goldberg1985}
Charles~H. Goldberg and Douglas~B. West.
\newblock Bisection of circle colorings.
\newblock {\em SIAM Journal on Algebraic Discrete Methods}, 6(1):93--106,
  January 1985.
\newblock \href {https://doi.org/10.1137/0606010} {\path{doi:10.1137/0606010}}.

\bibitem{Har-Peled2011}
Sariel Har-Peled.
\newblock {\em Geometric Approximation Algorithms}, volume 173 of {\em
  Mathematical Surveys and Monographs}.
\newblock {AMS}, June 2011.

\bibitem{He2019}
Jiafan He, Ariel~D. Procaccia, Alexandros Psomas, and David Zeng.
\newblock Achieving a fairer future by changing the past.
\newblock In {\em Proceedings of the Twenty-Eighth International Joint
  Conference on Artificial Intelligence}, IJCAI-2019, pages 343--349, Macao,
  China, August 2019. International Joint Conferences on Artificial
  Intelligence Organization.
\newblock \href {https://doi.org/10.24963/ijcai.2019/49}
  {\path{doi:10.24963/ijcai.2019/49}}.

\bibitem{JafarnejadGhomi2017}
Einollah Jafarnejad~Ghomi, Amir Masoud~Rahmani, and Nooruldeen Nasih~Qader.
\newblock Load-balancing algorithms in cloud computing: A survey.
\newblock {\em Journal of Network and Computer Applications}, 88:50--71, June
  2017.
\newblock \href {https://doi.org/10.1016/j.jnca.2017.04.007}
  {\path{doi:10.1016/j.jnca.2017.04.007}}.

\bibitem{Kash2014}
Ian Kash, Ariel~D. Procaccia, and Nisarg Shah.
\newblock No agent left behind: Dynamic fair division of multiple resources.
\newblock {\em Journal of Artificial Intelligence Research}, 51:579–603,
  November 2014.
\newblock \href {https://doi.org/10.1613/jair.4405}
  {\path{doi:10.1613/jair.4405}}.

\bibitem{Kraska2018}
Tim Kraska, Alex Beutel, Ed~H. Chi, Jeffrey Dean, and Neoklis Polyzotis.
\newblock The case for learned index structures.
\newblock In {\em Proceedings of the 2018 International Conference on
  Management of Data}, SIGMOD '18, pages 489--504. {ACM}, May 2018.
\newblock \href {https://doi.org/10.1145/3183713.3196909}
  {\path{doi:10.1145/3183713.3196909}}.

\bibitem{Mehrabi2021}
Ninareh Mehrabi, Fred Morstatter, Nripsuta Saxena, Kristina Lerman, and Aram
  Galstyan.
\newblock A survey on bias and fairness in machine learning.
\newblock {\em {ACM} Computing Surveys}, 54(6):1--35, July 2021.
\newblock \href {https://doi.org/10.1145/3457607} {\path{doi:10.1145/3457607}}.

\bibitem{Pessach2022}
Dana Pessach and Erez Shmueli.
\newblock A review on fairness in machine learning.
\newblock {\em {ACM} Computing Surveys}, 55(3):1--44, February 2022.
\newblock \href {https://doi.org/10.1145/3494672} {\path{doi:10.1145/3494672}}.

\bibitem{Sabek2022}
Ibrahim Sabek, Kapil Vaidya, Dominik Horn, Andreas Kipf, Michael Mitzenmacher,
  and Tim Kraska.
\newblock Can learned models replace hash functions?
\newblock {\em Proceedings of the {VLDB} Endowment}, 16(3):532--545, November
  2022.
\newblock \href {https://doi.org/10.14778/3570690.3570702}
  {\path{doi:10.14778/3570690.3570702}}.

\bibitem{Shahbazi2023}
Nima Shahbazi, Yin Lin, Abolfazl Asudeh, and {H.\,V.} Jagadish.
\newblock Representation bias in data: A survey on identification and
  resolution techniques.
\newblock {\em {ACM} Computing Surveys}, 55(13s):1--39, July 2023.
\newblock \href {https://doi.org/10.1145/3588433} {\path{doi:10.1145/3588433}}.

\bibitem{Shahbazi2024}
Nima Shahbazi, Stavros Sintos, and Abolfazl Asudeh.
\newblock {FairHash}: A fair and memory/time-efficient hashmap.
\newblock {\em Proceedings of the ACM on Management of Data}, 2(3):1--29, May
  2024.
\newblock \href {https://doi.org/10.1145/3654939} {\path{doi:10.1145/3654939}}.

\end{thebibliography}

\appendix

\section{Offline Algorithm is Optimal}
\label{sec:offline_optimal}
\begin{proposition}
    For every $k$, there exists a necklace that requires exactly $2(k-1)$ cuts.
\end{proposition}
\begin{proof}
    The $k=1$ case is trivial, so we assume $k>1$. Consider the necklace with $\frac{m}{2}$ red beads followed by $\frac{m}{2}$ blue beads. Suppose we have an optimal set of cuts of the necklace. The first interval in the resulting allocation consists of at most $\frac{m}{2k}$ red beads. Without loss of generality, assume this interval belongs to $A_1$. If it has exactly $\frac{m}{2k}$ beads, we move to the next interval (without loss of generality, owned by $A_2$). Else, we check if the first two intervals combined have more than $\frac{m}{2k}$ beads. If so, we can move the first cut to the right until the first interval has exactly $\frac{m}{2k}$ beads. To maintain the fairness constraint, we reassign all the red beads of $A_1$ in the rest of the necklace to $A_2$ (without adding any cuts). We then move to the next interval. Instead, if the first two intervals combined have at most $\frac{m}{2k}$ beads, we remove the cut between them and add one cut within another interval with red beads belonging to $A_1$, maintaining the overall number of cuts.

    We repeat this process from both the left (red) and right (blue) ends of the necklace until we reach the center interval, which consists of $\frac{m}{2k}$ red beads and $\frac{m}{2k}$ blue beads. We have exactly $2(k-1)$ cuts, and since the number of cuts remained constant throughout this procedure, the original optimal set of cuts was of size $2(k-1)$.
\end{proof}

\section{NP-completeness of \MinNodeMaxFlow{}}
\label{sec:minnodenaxflow}

\begin{proof}[Proof of Proposition~\ref{prop:minnodemaxflow}]
    First, we show that \MinNodeMaxFlow{} is in NP. We are given a flow network and a limit $\kappa$ on the number of active nodes as input and a proposed flow as a certificate. We can compute the value of a max flow on the graph (in polynomial time). Then we can easily verify whether the proposed flow has that value and satisfies the flow constraints and whether the number of active nodes is at most $\kappa$. \MinNodeMaxFlow{} is in NP.
    
    Next, we show that \MinNodeMaxFlow{} is NP-hard, via reduction from \textsf{Vertex Cover}. Given an undirected graph $\Gamma = (V_\Gamma, E_\Gamma)$ and parameter $\kappa$, we want to determine whether $\Gamma$ has a vertex cover of size at most $\kappa$. We construct an auxiliary, directed graph $\Gamma' = (V_{\Gamma'}, E_{\Gamma'})$ containing the same vertices in addition to a source node $\sigma$, a sink node $\tau$, and an extra node $\overline{uv}$ for every edge $(u,v) \in E_\Gamma$. We replace each edge $(u,v) \in E_\Gamma$ with two infinite-capacity edges: $(\overline{uv}, u)$ and $(\overline{uv}, v)$. We add a unit-capacity edge from $\sigma$ to each extra node $\overline{uv}$ and an infinite-capacity edge from each original node $u$ to $\tau$.
    
    If there is a vertex cover of size $\kappa$ in $\Gamma$, then there is a max flow on $\Gamma'$ with $|E_\Gamma| + \kappa$ active nodes. Conversely, suppose there is a max flow on $\Gamma'$ with $\kappa$ active nodes. If any node $\overline{uv}$ is sending fractional flow to each of its neighbors, we can reroute its flow to go to just one neighbor without changing the number of active nodes. We would need to do this at most $|E_\Gamma|$ times, so we can always efficiently find an integral solution, given the initial solution. Once we obtain an integral solution, we have a corresponding vertex cover of size $\kappa$. Therefore, \MinNodeMaxFlow{} is NP-hard.
\end{proof}

\section{Correctness and Time Complexity of \ApproxStatic{}}
\label{sec:approxstatic}

\begin{definition}[$\eps$-sample~{\cite[Chapter~5]{Har-Peled2011}}]
	Let $(X,\mathcal{R})$ be a set system. For any $\eps \in [0,1]$, a subset $C \subseteq X$ is an \emph{$\eps$-sample} for $(X,\mathcal{R})$ if for every $\rho \in \mathcal{R}$, we have
	\[\biggl|\frac{|\rho|}{|X|} - \frac{|C \cap \rho|}{|C|}\biggr| \leq \eps \,.\]
\end{definition}

\begin{proof}[Proof of Theorem~\ref{thm:approxstatic}]
We can construct $\tree$ in $O(m)$ time. We can retrieve an element from $\tree$ in $O(\log m)$ time, so we can construct the set $\samples^j$ in $O(|\samples^j|\log m)=O\bigl(k^2 2^{2k}\eps^{-2}(\log m)^2\bigr)$ time. Executing the exact offline algorithm for one iteration takes $O(|\samples^{{j}}|) = O\bigl(k^2 2^{2k}\eps^{-2}\log m\bigr)$ time. We run the subroutine for $k$ iterations, so the total running time of \ApproxStatic{} is $O\bigl(m+k^3 2^{2k}\eps^{-2}(\log m)^2\bigr)$.
It is clear that, by construction, the algorithm produces at most $2(k-1)$ cuts. What remains is to show that the resulting assignment of beads is approximately fair.

For any subset of beads $X$ in the necklace, consider the set system $(X, \mathcal{R}_S)$ of VC dimension $2$, where $\mathcal{R}_S$ is the family of all possible intervals (sets of contiguous beads) in $S$. Then, with probability at least $1 - \varphi$, a uniform random subset $C \subseteq X$ of cardinality $O\bigl(\eps^{-2}\log\bigl(\varphi^{-1}\bigr)\bigr)$ is an $\eps$-sample for $(X, \mathcal{R}_S)$~\cite{Anthony1999}.

Let $\epss = \frac{\eps}{2^k}$. By definition, the probability that $\samples_i^j$ is a $\frac{\epss}{k-j+1}$-sample is at least $1-\frac{1}{2km}$. We have that in the $j$th iteration, both $\samples_1^j$ and $\samples_2^j$ are $\frac{\epss}{k-j+1}$-samples for $(S_1\setminus \intervals^{j-1}$, $\mathcal{R}_S)$ and $(S_2\setminus \intervals^{j-1}, \mathcal{R}_S)$, respectively, with probability at least $1-\frac{1}{km}$. Via a union bound over the iterations, with probability at least $1 - \frac{1}{m}$, we have that $\samples_i^j$ is a $\frac{\epss}{k-j+1}$-sample for $(S_i\setminus \intervals^{j-1}, \mathcal{R}_S)$ for all $j \in [k]$ and $i \in \set{1, 2}$.

We now prove by induction that for each iteration $j\in[k]$ and color $i\in\set{1, 2}$, we have
\[\frac{1-(2^j-1)\epss}{k} \leq \frac{|S_i \cap I_j|}{|S_i|} \leq \frac{1+(2^j-1)\epss}{k}\]
if $\samples_i^j$ is an $\frac{\epss}{k-j+1}$-sample for $S_i\setminus \intervals^{j-1}$.
First we consider the base case: $j=1$. By definition, in the first iteration, $\samples_i^1$ is a $\frac{\epss}{k-j+1}$-sample for $S_i$. By the correctness of the exact offline algorithm on $\samples^1$, for each $i \in \set{1, 2}$, we have that $|\samples_i^1 \cap I_1|/|\samples_i^1| = 1/k$. Thus, we have \[\frac{1-\epss}{k} \leq \frac{|S_i\cap I_1|}{|S_i|} \leq \frac{1+\epss}{k} \,.\]

For the inductive step, assume that the claim holds for all $j < \zeta$. We will show that it holds for $j = \zeta$. By the correctness of the exact offline algorithm on $\samples^1$ we have that $|\samples_i^\zeta\cap I_\zeta|/|\samples_i^\zeta| = 1/(k-\zeta+1)$. Since $\samples_i^\zeta$ is a $\frac{\epss}{k-\zeta+1}$-sample for $S_i\setminus\intervals^{\zeta-1}$, we have
\[\frac{1-\epss}{k-\zeta+1} \leq \frac{|(S_i\setminus \intervals^{\zeta-1})\cap I_\zeta|}{|S_i\setminus\intervals^{\zeta-1}|}\leq \frac{1+\epss}{k-\zeta+1} \,.\]
Notice that $(S_i\setminus\intervals^{\zeta-1})\cap I_\zeta=S_i\cap I_\zeta$, so we can rewrite this as
\[\frac{1-\epss}{k-\zeta+1}|S_i\setminus\intervals^{\zeta-1}| \leq |S_i\cap I_\zeta|\leq \frac{1+\epss}{k-\zeta+1}|S_i\setminus\intervals^{\zeta-1}| \,.\]
We first prove the left inequality of the claim. 
\begin{align*}
    |S_i\cap I_\zeta|
    &\geq \frac{1-\epss}{k-\zeta+1}\bigl\lvert S_i \setminus \intervals^{\zeta-1}\bigr\rvert \\
    &= \frac{1-\epss}{k-\zeta+1} \biggl(m_i-\sum_{j=1}^{\zeta-1}|S_i\cap I_j|\biggr) \\
    &\geq \frac{1-\epss}{k-\zeta+1} \Biggl(1-\sum_{j=1}^{\zeta-1}\frac{1+\bigl(2^j-1\bigr)\epss}{k}\Biggr) m_i \\
    &= (1-\epss) \biggl(1-\frac{\epss}{k-\zeta+1}\sum_{j=1}^{\zeta-1}\bigl(2^j-1\bigr)\biggr) \frac{m_i}{k} \\
    &= (1-\epss) \Biggl(1-\frac{\bigl(-\zeta+2^\zeta-1\bigr) \epss}{k-\zeta+1}\Biggr) \frac{m_i}{k} \\
    &\geq (1-\epss) \bigl(1-\bigl(-\zeta+2^\zeta-1\bigr)\epss\bigr) \frac{m_i}{k} \\
    &\geq \bigl(1-\epss+\zeta\epss-2^\zeta\epss+\epss\bigr)\frac{m_i}{k} \\
    &\geq \bigl(1-\bigl(2^\zeta-1\bigr)\epss\bigr)\frac{|S_i|}{k}
\end{align*}
Similarly, we prove the right inequality.
\begin{align*}
    |S_i\cap I_\zeta|
    &\leq \frac{1+\epss}{k-\zeta+1} \bigl\lvert S_i \setminus \intervals^{\zeta-1}\bigr\rvert \\
    &= \frac{1+\epss}{k-\zeta+1}\biggl(m_i-\sum_{j=1}^{\zeta-1}|S_i\cap I_j|\biggr) \\
    &\leq \frac{1+\epss}{k-\zeta+1} \Biggl(1-\sum_{j=1}^{\zeta-1}\frac{1-\bigl(2^j-1\bigr)\epss}{k}\Biggr) m_i \\
    &= (1+\epss)\Biggl(1+\frac{\bigl(-\zeta+2^\zeta-1\bigr) \epss}{k-\zeta+1} \Biggr) \frac{m_i}{k} \\
    &\leq (1+\epss) \bigl(1+\bigl(-\zeta+2^\zeta-1\bigr)\epss\bigr) \frac{m_i}{k} \\
    &= \bigl(1+2^\zeta\epss-\zeta\epss + \bigl(2^\zeta-\zeta-1\bigr)\epss^2\bigr)\frac{m_i}{k} \\
    &< \bigl(1+2^\zeta\epss-\zeta\epss + \epss\bigr)\frac{m_i}{k} \\
    &\leq \bigl(1+\bigl(2^\zeta-1\bigr)\epss\bigr) \frac{|S_i|}{k}
\end{align*}
By induction, the claim holds. Finally, substituting $\epss = \frac{\eps}{2^k}$, we have
\[\frac{(1-\eps)|S_i|}{k} < |S_i\cap I_j| < \frac{(1+\eps)|S_i|}{k} \,.\qedhere\]
\end{proof}

\end{document}